\documentclass[11p,reqno]{amsart}

\pdfoutput=1

\usepackage{amssymb}
\usepackage{amsthm}
\usepackage{amsfonts}
\usepackage{amsmath}
\usepackage{tikz,tikz-cd}
\usetikzlibrary{math}
\usetikzlibrary{decorations}
\usepackage{pmboxdraw}
\usepackage{verbatim} % for comment
\usepackage{graphicx}
\usepackage{color}
\usepackage{cite}
\usepackage[normalem]{ulem}
\usepackage{subcaption}
\usepackage{todonotes}
\usepackage{bbm}
\usepackage{tcolorbox}
\usepackage{marginnote}
\tcbuselibrary{breakable}
\usepackage{subcaption}

\numberwithin{equation}{section}
\newtheorem{theorem}{Theorem}[section]

\newtheorem{coro}[theorem]{Corollary}

\newtheorem{proposition}[theorem]{Proposition}

\theoremstyle{remark}
\newtheorem{remark}[theorem]{Remark}

\begin{document}

\title[]{On the superintegrability of the Gaussian $\beta$ ensemble \\ and its $(q,t)$ generalisation}

%\subjclass[2020{15B5}

\date{}

%%%%%%%%%%%%%%%%%%%%%%%%%%%%% author %%%%%%%%%%%%%%%%%%%%%%%%%%%%
\author{Sung-Soo Byun}
\address{Department of Mathematical Sciences and Research Institute of Mathematics, Seoul National University, Seoul 151-747, Republic of Korea}
\email{sungsoobyun@snu.ac.kr}

\author{Peter J. Forrester}
\address{School of Mathematical and Statistics, The University of Melbourne, Victoria 3010, Australia}
\email{pjforr@unimelb.edu.au} 
%%%%%%%%%%%%%%%%%%%%%%%%%%%%% author %%%%%%%%%%%%%%%%%%%%%%%%%%%%

\thanks{Sung-Soo Byun was supported by the National Research Foundation of Korea grant (RS-2023-00301976, RS-2025-00516909). Peter Forrester was supported by the Australian Research Council discovery project grant DP250102552.
}

%\dedicatory{}

%\keywords{ }

\begin{abstract}
In the present context, superintegrability is a property of certain probability density functions coming from matrix models, which relates to the average over a distinguished basis of symmetric functions, typically the Jack or Macdonald polynomials. It states that the average can be computed according a certain combination of those same polynomials, now specialised by specific substitutions when expressed in terms of the power sum basis. For a particular $(q,t)$ generalisation of the Gaussian $\beta$ ensemble from random matrix theory, known independently from the consideration of certain integrable gauge theories, we use results developed in a theory of multivariable Al-Salam and Carlitz polynomials based on Macdonald polynomials to prove the superintegrability identity. This then is used to deduce a duality formula for these same averages, which in turn allows for a derivation of a functional equation for the spectral moments.
\end{abstract}

\maketitle

\section{Introduction}

Familiar in random matrix theory is the eigenvalue probability density function (PDF) on $(x_1,x_2,\dots,x_N) \in\mathbb R^N$,
\begin{equation}\label{1.1a}
{1 \over Z_N^{\rm G} } \prod_{1 \le i < j \le N} | x_i - x_j|^\beta \prod_{l=1}^N e^{-\beta x_l^2/2},
\end{equation}
with $\beta > 0$ and where $Z_N^{\rm G}$ is the normalisation
\begin{equation}\label{1.1b}
Z_N^{\rm G} = \beta^{-N/2 - N \beta (N - 1)/4}
(2 \pi)^{N/2} \prod_{j=0}^{N-1} { \Gamma(1 + (j + 1) \beta/2) \over \Gamma(1 + \beta/2) }.
\end{equation}
This is said to define the Gaussian $\beta$ ensemble; see e.g.~\cite[Ch.~1]{Fo10}. An overview of settings in random matrix theory giving rise to (\ref{1.1a}) is given in \S \ref{S1.1}.

Let $\langle \cdot \rangle^{{\rm G\beta E}_N}$ denote an average with respect to (\ref{1.1a}), and let $P_\kappa^{(2/\beta)}(\mathbf x)$ with $\mathbf x = (x_1,\dots,x_N)$ denote the symmetric Jack polynomial; this class of special function is to be discussed further in \S \ref{S1.2}. For now we remark that the Jack symmetric polynomials are labelled by a partition $\kappa = (\kappa_1,\dots,\kappa_N)$, depend on $\beta$, and form a basis for analytic symmetric functions of $N$ variables.
The property of the Gaussian $\beta$ ensemble of most interest to us in the present work is the explicit formula for the $\langle \cdot \rangle^{{\rm G\beta E}_N}$ average over $P_\kappa^{(2/\beta)}$,
\begin{equation}\label{4.2}
2^{| \kappa |/2} \langle  P_\kappa^{(2/\beta)}(\mathbf x) \rangle^{{\rm G\beta E}_N} 
 =
P_\kappa^{(2/\beta)}(\mathbf x) \Big |_{\mathbf x = \mathbf 1}
{ [p_2^{|\kappa|/2}] P_\kappa^{(2/\beta)}(\mathbf x) \over [p_1^{|\kappa|}] P_\kappa^{(2/\beta)}(\mathbf x )} =
P_\kappa^{(2/\beta)}(\{ p_k = N \})
{ [p_2^{|\kappa|/2}] P_\kappa^{(2/\beta)}(\mathbf x) \over [p_1^{|\kappa|}] P_\kappa^{(2/\beta)}(\mathbf x )},
\end{equation}
where $|\kappa|$ is assumed even, or otherwise the left hand side vanishes;
references are deferred until the next paragraph.
Here $p_j := \sum_l x_l^j$ is the $j$-th power sum, $\mathbf x= \mathbf 1$ denotes that $x_j = 1$ for each $j=1,\dots,N$, the notation $[p_j^r] P_\kappa^{(2/\beta)}$ denotes the coefficient of $p_j^r$ in the power sum expansion of $P_\kappa^{(2/\beta)}$, and the notation 
$P_\kappa^{(2/\beta)}(\{ p_k = N \})$ denotes that each $p_k$ is to be replaced by $N$ in the same power sum expansion.

Results from Jack polynomial theory give
\begin{equation}\label{4.3}
{P_\kappa^{(2/\beta)}(\mathbf x) |_{\mathbf x = \mathbf 1}
\over [p_1^{|\kappa|}] P_\kappa^{(2/\beta)}(\mathbf x)} 
= (2/\beta)^{|\kappa|} [\beta N/2]_\kappa^{(2/\beta)},
\end{equation}
where $|\kappa|:= \sum_{j=1}^N \kappa_j$  and
\begin{equation}\label{4.3m}
\quad[u]_\kappa^{(\alpha)} := \prod_{l=1}^N {\Gamma(u - (l-1)/\alpha + \kappa_l) \over \Gamma(u - (l - 1)/\alpha)};
\end{equation}
see e.g.~\cite[Eqns.~(12.104), (12.105) and Prop.~12.6.7]{Fo10}). With (\ref{4.3m}) substituted in (\ref{4.3}), and this in turn substituted in the first equality of(\ref{4.2}), the resulting identity is
equivalent to a conjecture of Goulden and Jackson \cite{GJ97}, first proved by Okounkov \cite{Ok97} and subsequently, incorporating  different strategies, by Dumitriu \cite{Du03} and Desrosiers \cite{De09}.
In the special case $\beta = 2$, when the Jack polynomial is equal to the more familiar Schur polynomial, the identity in the structured form given in the second equality of (\ref{4.2}) --- involving only quantities on the right hand side that can be directly related to the same Jack polynomial as appearing on the 
left hand side when expressed in the power sum basis --- was derived in a work by Itzykson and Zuber \cite{IZ90}. 
There are now several known identities  exhibiting such a structured form --- see in particular (\ref{Mq}) below --- which in recent literature \cite{MM21,MM22,MM24} has been termed a superintegrability property. We remark that a proof of (\ref{4.2}) exhibiting directly the superintegrability form has been given recently in  \cite{WLZZ22,WZZZ22}; see also \S \ref{S2.1} below.

From various viewpoints there is interest in $q$ (or more generally $(q,t)$) generalising (\ref{1.1a}), and its analogue for unitary matrices, the circular $\beta$ ensemble. Thus one finds motivations from considerations of beta integrals \cite{As80,LSYF22}, Macdonald  polynomial theory \cite{Ko96,BF98,Ma07,Cu18,Ol21}, $q$-Virasoro constraints \cite{SKAO96,NZ17,CLPZ19,LPSZ20,CZ22,KK22}, the computation of moments and combinatorics \cite{Wi12,Ve14,MPS20,Co21,FLSY23,BFO24} (these latter references specialise to the case $\beta = 2$), 
inhomogeneous tiling-type models and Gaussian free field fluctuations
\cite{FR02b,DK19,GH22,LWYZ24},
and gauge theories \cite{AOSV05,ST13,Sh15,CLZ20}. Here our list of references, while comprehensive, is far from complete --- many more could be given. 
The first step in $q$ generalising (\ref{1.1a}) is to identify the appropriate replacement of the Gaussian weight. For this, we introduce a further parameter $a < 0$, and define 
 \begin{equation}\label{5.2c}
 w_U^{(a)}(x;q) := {(qx;q)_\infty (qx/a;q)_\infty \over (q;q)_\infty (a;q)_\infty (q/a;q)_\infty},
 \end{equation}
 known from the theory of the Al-Salam and Carlitz $q$-orthogonal polynomials \cite{KLS10} (in the case $a=-1$ (\ref{5.2c}) is the weight function for the discrete $q$-Hermite polynomials; 
 the notation used on the right hand side of (\ref{5.2c}) is given in \S \ref{S1.3}, along with justifying it as a
 generalisation of the Gaussian weight.
Next, with $m \in \mathbb Z^+$, $t:= q^m$, following \cite{BF98} we appropriately generalise the  product of differences in (\ref{1.1a}), and introduce as a $(q,t)$ analogue   the functional form
 \begin{equation}\label{5.3}
 {1 \over \mathcal N_{0,N}(a;q,t)} \prod_{l=1}^N  w_U^{(a)}(x_l;q) \prod_{p=-(m-1)}^m \prod_{1 \le i < j \le N} (x_i - q^p x_j),
  \end{equation}
  supported on the $q$-lattice $\{1,q,q^2,\dots \} \cup \{a,aq,aq^2,\dots \}$ for each $x_l$. With the requirement that $a<0$ this can be checked to be non-negative. While it is not symmetric in $\{x_1,\dots,x_N\}$, when averaged against symmetric functions it can be replaced by a modification,
  namely (\ref{5.3x}) below, which is symmetric. The normalisation has the evaluation \cite[Eq.~(4.27)]{BF98}
  \begin{equation}\label{5.3a}
 \mathcal N_{0,N}(a;q,t) =  (1-q)^N (-a)^{m N (N - 1)/2}
 t^{m \binom{N}{3} - \frac{m-1}{2}  \binom{N}{2}}
 \prod_{l=1}^N {(q;q)_{ml} \over (q;q)_m};
 \end{equation}
  cf.~(\ref{1.1b}).

Let $\langle \cdot \rangle^{(a,q,t)}_N$ denote an average with respect to (\ref{5.3}). As a $(q,t)$ generalisation of the Jack polynomials introduce the Macdonald polynomials $P_\kappa(\mathbf x; q,t)$ \cite{Ma95} --- for more on these see \S \ref{S1.3} below. Then as a   $(q,t)$ generalisation of (\ref{4.3}) the recent works \cite{MPS18},
   \cite{MPS20} and \cite{MM22} (after minor modification to match our setting) give as a conjecture the following superintegrability identity (characteristic feature being that the average of the polynomial is equal to an expression involving the same polynomial in the power sum basis), which we are able to prove. 

\begin{theorem}\label{T1}
We have
\begin{equation}\label{Mq}
 \langle  P_\kappa(\mathbf x;q,t) \rangle^{(a,q,t)}_N 
 =
P_\kappa\Big ( \Big \{ p_k = {1 - t^{kN} \over 1 - t^k}
\Big \}; q,t \Big )
{  P_\kappa \Big ( \Big \{ p_k = {1 + a^k  \over 1 - t^k}
\Big \};q,t \Big ) \over P_\kappa \Big ( \Big \{ p_k = {1 \over 1 - t^k}
\Big \}; q,t \Big )},
\end{equation}
where on the right hand side the Macdonald polynomials are viewed as being specified in the power sum basis, with each $p_k$ therein then substituted as specified.
\end{theorem}

Our proof is given in \S \ref{S3.2s}. There we show too how to give meaning to $\langle \cdot \rangle^{(a,q,t)}_N$ beyond the case $t = q^m$, $m \in \mathbb Z^+$, when the underlying probability density function is specified by (\ref{5.3});
see Remark \ref{RE1}. Moreover, a structural feature of the right hand side of (\ref{Mq})
beyond that already highlighted as identifying it as a superintegrability identity,  is that it is a rational function in $(a,q,t,t^N)$, and so is well defined independent of any assumed relation between $t$ and $q$.
% We will see later in (\ref{4.3q}) below that on the right hand side of (\ref{T1}) the ratio of the first factor and the denominator can be further simplified analogous to \eqref{4.3}. However, we will see for the purposes of several corollaries and applications that there can be an advantage in keeping (\ref{Mq}) as is.

We next proceed to consider some applications. 
The first is a duality formula for Macdonald polynomial averages, which generalises a known duality for the Gaussian $\beta$ ensemble (see \S 3.1).

\begin{coro}\label{C1q}
Let $\kappa$ be a partition, and denote by $\ell(\kappa)$ the number of nonzero parts in $\kappa$. Let $\kappa'$ denote the partition conjugate to $\kappa$ (see text below (\ref{h1})).
For $\ell(\kappa) \le N$ and $\ell(\kappa') \le M$ we have
\begin{equation}\label{GFq}
\Bigg \langle {P_\kappa(\mathbf x;q,t) \over 
P_\kappa(1,t,\dots,t^{N-1};q,t)} 
\Bigg \rangle^{(a,q,t)}_N
=
\Bigg \langle {P_{\kappa'}(\mathbf x;t^{-1},q^{-1}) \over 
P_{\kappa'}(1,q^{-1},\dots,q^{-M+1};t^{-1},q^{-1})} 
\Bigg \rangle^{(a,t^{-1},q^{-1})}_M.
\end{equation}
\end{coro}

The second of our applications of (\ref{dd1}), which more directly can be considered as a consequence of Corollary \ref{C1q}, is the derivation of a
certain functional equation of the type referred to in 
\cite[\S 2.6.1]{CZ22}  as a Langland's duality, for the moments
 \begin{equation}\label{mbq}
m_{N,p}^{(a,q,t)}= M_p(a,q,t,t^N) := 
\Big \langle \sum_{l=1}^N x_l^{p} \Big \rangle^{(a,q,t)}_N.
\end{equation}
On this we are guided by the known functional equation for the moments of the Gaussian $\beta$ ensemble (see \S \ref{S4.3}).

\begin{coro}\label{P5}
The moments (\ref{mbq}) satisfy the functional relation
 \begin{equation}\label{mbq1} 
 M_{p}(a,q,t, u) = 
 -  q^{-p}  \bigg ( {1 - q^p \over 1 - t^p} \bigg ) M_{p}(a,t^{-1},q^{-1},u),
 \end{equation}
 where $u := t^N$.
\end{coro}

We conclude this Introduction with an outline of the remainder of the paper. In \S \ref{S1.1} we give some context to the Gaussian $\beta$ ensemble in random matrix theory.  In \S \ref{S1.2} the moments of the Gaussian $\beta$ ensemble are introduced as the power sum average, with motivations coming from the viewpoint of a certain topological expansion, and also that of the spectral density. In this section Jack polynomials are introduced, ready for their appearance in the averages of \S \ref{S2.1}. Section \ref{S1.3} gives motivation for the introduction of (\ref{5.3}) as the $(q,t)$ analogue of (\ref{1.1a}), and furthermore introduces Macdonald polynomials as the $(q,t)$ generalisation of the Jack polynomials. In \S \ref{S2.1}
the generalised hypergeometric function  ${\vphantom{\mathcal F}}_0^{\mathstrut}\mathcal F_0^{(\alpha)}(\mathbf x; \mathbf y)$ is introduced, and following earlier literature is used to give a derivation of  the Gaussian $\beta$ ensemble/Jack polynomial analogue of Theorem \ref{T1}. It is shown how the key steps therein can be $(q,t)$ generalised in \S \ref{S3.2s}, allowing for a 
proof of Theorem \ref{T1}.
Proofs of Corollaries \ref{C1q} and \ref{P5} are given in \S \ref{S3.1b}
and \S \ref{S4.4} respectively, with the preceding subsections containing proofs of their known analogues for the Gaussian $\beta$ ensemble, which we use for guidance. Also given in \S \ref{S3.1b}, as an application of 
Corollary \ref{C1q},
is the evaluation of
$\langle \prod_{l=1}^N (u - x_l) \rangle^{(a,q,t)}_N$,
which has the interpretation as the averaged characteristic polynomial.
In Remark \ref{R4.3} of this section we discuss challenges associated with using the duality formula for the purpose of computing large $N$ asymptotics. Additional material of 
\S \ref{S4.4} includes a
listing of the explicit forms of the moments $M_p(a,q,t,u)$ for $p=1,2,3$.

\section{Preliminaries relating to the ensembles and the superintegrability identities}
\subsection{Settings of the Gaussian $\beta$ ensemble}\label{S1.1}

In random matrix theory (\ref{1.1a}) is recognised as the eigenvalue PDF of Hermitian Gaussian random matrices $G$ chosen with  matrix distribution proportional to $e^{-\beta {\rm Tr} \, G^2/2}$, and with elements that are real $(\beta = 1)$, complex $(\beta = 2)$ or real quaternion $(\beta = 4)$; see e.g.~\cite[Ch.~1]{Fo10}. 
Such random matrices form the Gaussian orthogonal, unitary and symplectic random matrix ensembles respectively. Wigner \cite{Wi57a} (reprinted in \cite{Po65}) observed that up to the normalisation constant, the case $\beta = 1$ of (\ref{1.1a}) is identical to the Boltzmann factor $e^{-\beta W}$ of the classical gas on a line with potential
\begin{equation}\label{2.1b}
W = {1 \over 2} \sum_{l=1}^N x_l^2 - \sum_{1 \le j < k \le N}
\log | x_j - x_k |,
\end{equation}
and with $\beta$ interpreted as the dimensionless inverse temperature. This gives a model for the eigenvalues with PDF (\ref{1.1a}) and general $\beta > 0$ as a so-called log-gas, which was a viewpoint particularly prevalent in the work subsequent to Wigner by Dyson \cite{Dy62e}.

A realisation of (\ref{1.1a}) for general $\beta > 0$ as an eigenvalue PDF is given by a class of real symmetric tridiagonal matrices \cite{DE02}.
These tridiagonal matrices are specified by having independent normals on their diagonal, and the sequence of independent but non-identical square rooted gamma random variables with
shape parameter $\{(N-k) \}_{k=1,\dots,N-1}$ on their (independent) off diagonal. Also, a construction of a recursively defined random matrix ensemble which realises (\ref{1.1a}) for general $\beta > 0$ has been given in \cite{Fo13}.

In \cite{Dy62b} Dyson let the independent entries of the classical Gaussian Hermitian models evolve according to Brownian motion, and deduced that the corresponding eigenvalue PDF,
$p_\tau$ say, satisfies the Fokker-Planck equation
\begin{equation}\label{2.1a}
{\partial p_\tau \over \partial \tau} = \mathcal L p_\tau, \qquad  \mathcal L = {1 \over \beta} \sum_{j=1}^N {\partial^2 \over \partial x_j^2} + \sum_{j=1}^N
{\partial \over \partial x_j} {\partial W \over \partial x_j}.
\end{equation}
This has the equilibrium, $\tau \to \infty$, solution (\ref{1.1a}) for all $\beta > 0$.
Sutherland  \cite{Su71} observed that conjugation by $e^{-\beta W/2}$ transforms (\ref{2.1a}) to the Schr\"odinger operator for the Calogero-Sutherland many body system specified by   \cite{Ca69}
\begin{equation}\label{2.0}
\mathcal H = - \sum_{j=1}^N {\partial^2 \over \partial x_j^2}
+ {\beta^2 \over 4} \sum_{j=1}^N x_j^2 + \beta (\beta/2 - 1
) \sum_{1 \le j < k \le N} {1 \over (x_j - x_k)^2}.
\end{equation}
Thus (\ref{1.1a}) also has the interpretation as the absolute valued squared of the ground state for the quantum many body system.

% Thus let $X$ be a member of one of the so defined parameter independent ensembles. Set 
% $G = |1 - e^{-2 \tau} |^{1/2} X + 
% e^{-\tau} X_0$, where $X_0$ is a fixed Hermitian matrix with elements from the same number field as $X$. Then the corresponding eigenvalue PDF can be shown to satisfy 
% (\ref{2.1a}) \cite{Dy62b}, \cite[Ch.~11]{Fo10}. The equilibrium, $\tau \to \infty$ solution of (\ref{2.1a}) for all $\beta > 0$ is (\ref{1.1a}). It is furthermore true that $p_\tau$ itself can be realised as an eigenvalue probability density of a recursively defined random matrix ensemble for general $\beta > 0$ \cite{Fo13}. 

\subsection{Preliminaries on Gaussian 
$\beta$ ensemble averages  and Jack polynomials}\label{S1.2}

A celebrated application of the $\beta =2$ complex Hermitian  random matrix ensemble relating to (\ref{1.1a}) (i.e.~the Gaussian unitary ensemble) is to the enumeration of certain maps on Riemann surfaces \cite{BIPZ78}. Using the ideas of this work, the question of counting the number of pairings $\mathcal E_g(p)$ of the sides of a $2p$-gon giving rise to a surface of genus $g$ was addressed in \cite{HZ86}. With $\langle \cdot \rangle^{\rm GUE}$ denoting an average with respect to GUE random matrices, and $m_{N,2p}^{\rm GUE} :=
\langle {\rm Tr} \, G^{2p} \rangle^{\rm GUE}$, it was shown that 
\begin{equation}\label{3.1}
2^{p} m_{N,2p}^{\rm GUE} = \sum_{g=0}^{ \lfloor (p+1)/2 \rfloor}
\mathcal E_g(p) N^{p+1-2g}.
\end{equation}
This is said to define a topological expansion.
From (\ref{3.1}), various characterisations of $\{ \mathcal E_g(p) \}$ were given, including that of the solution of a certain two parameter recurrence relation. Most recently a finite summation formula involving $2g+1$ terms, and with the Stirling numbers of the first kind as a factor in the summand, has been given for $\mathcal E_g(p)$ \cite[Eq.(1.27)]{BFO24}.

There is a relation between $m_{N,2p}^{\rm GUE}$ and the $\beta =2$ case of the superintegrability identity (\ref{4.2}). As previously remarked, the Jack polynomials in this case reduce to the Schur polynomials, $s_\kappa(\mathbf x)$ say. Particular Schur polynomials can be used to compute the power sums according to \cite{Ma95}
\begin{equation}\label{3.5}
\sum_{l=1}^N x_l^j = \sum_{r=0}^{{\rm min} \{j-1,N-1\} }
(-1)^r s_{(j-r,(1)^r)}(\mathbf x).
\end{equation}
Here $(1)^r$ denotes that the part 1 is repeated $r$ times.
It is shown in 
\cite{IZ90} that use of (\ref{4.2}) in (\ref{3.5}) implies the alternating sum formula
\begin{align}\label{3.8}
{2^{p} \over (2p-1)!!} m_{N,2p}^{\rm GUE} & =
\sum_{s=0}^{p-1} (-1)^s \binom{p-1}{s} 
\bigg \{  \binom{N+2p-2s-1}{2p} +
\binom{N+2p-2s-2}{2p} \bigg \} \nonumber \\
& = {1 \over 2} [y^{p+1}] \bigg ( {(1 + y)^N \over (1 - y)^N} \bigg ),
\end{align}
where the second equality reclaims a result first  given in \cite[\S 4]{HZ86} in the context of the topological expansion (\ref{3.1}).

There is an independent reason to be interested in $m_{N,2p}^{\rm GUE}$, which in fact extends to the general $\beta > 0$ analogue of this average. Thus upon noting that
in terms of the eigenvalues, ${\rm Tr} \, G^{2p} = \sum_{l=1}^N x_l^{2p}$, one sees that
\begin{equation}\label{3.2}
m_{N,2p}^{\rm GUE} = \int_{\mathbb R} x^{2p} \rho_N^{\rm GUE}(x) \, dx,
\end{equation}
where $\rho_N^{\rm GUE}(x)$ is the spectral density for the GUE. Hence the $m_{N,2p}^{\rm GUE}$ can be viewed as spectral moments. In fact as noted and made use of by Wigner \cite{Wi58}, knowledge of the leading term in (\ref{3.1}) being given by $\mathcal E_0(p) = C_p$, where $C_p := {1 \over p + 1} \binom{2p}{p}$ is the $p$-th Catalan
number, is sufficient to compute the global scaled density.
This is defined as the limiting eigenvalue density of ${1 \over \sqrt{N}} G$ for $G$ a GUE matrix. 
Thus one has
\begin{equation}\label{3.3}
\lim_{N \to \infty} {1 \over \sqrt{N}} \rho_N^{\rm GUE}( \sqrt{N}x) = {\sqrt{4 - x^2} \over 2 \pi} \mathbbm{1}_{(-2,2)}(x),
\end{equation}
where the functional form in this expression is known as the Wigner semi-circle law \cite{AGZ09,PS11}.

We now present the definition of the Jack polynomial as appearing in the superintegrability identity (\ref{4.2}), and some properties which are required in our subsequent working.
While the Schur polynomials permit a determinantal
expression (see e.g.~\cite[Eq.~(10.16)]{Fo10})
% \begin{equation}\label{3.4}
% s_\kappa(\mathbf x) = {\det [ x_k^{N-j+\kappa_j}]_{j,k=1,\dots,N} \over \det [ x_k^{N-j}]_{j,k=1,\dots,N}}, \qquad \mathbf x = (x_1,\dots,x_N),
%\end{equation}
there is no such explicit formula for the Jack polynomials.
As a characterisation of the Jack polynomials (see e.g.~\cite[\S 12.6]{Fo10}), one first 
requires that when expanded in terms of the monomial symmetric polynomials $\{ m_\mu(\mathbf x) \}$ they have the triangular structure
\begin{equation}\label{4.0}
P_\kappa^{(\alpha)}(\mathbf x) = m_\kappa(\mathbf x) +
\sum_{\mu < \kappa} c_{\kappa, \mu}^{(\alpha)} 
m_\mu (\mathbf x),
\end{equation}
for certain $N$ independent coefficients $\{ c_{\kappa, \mu}^{(\alpha)} \}$. The notation $\mu < \kappa$ denotes the partial order on partitions defined by the requirement that $\sum_{i=1}^s \mu_i \le \sum_{i=1}^s \kappa_i$, for each $s=1,\dots,N$ and with equality for $s=N$. Then each 
$P_\kappa^{(\alpha)}(\mathbf x)$ is uniquely determined as the polynomial eigenfunction of the differential operator
\begin{equation}\label{4.1}
\sum_{j=1}^N x_j^2 {\partial^2 \over \partial x_j^2} +
{1 \over \alpha} \sum_{  \substack{i,j=1 \\j \ne i }}^N {1 \over x_i - x_j} \Big ( x_i^2 {\partial \over \partial x_i} -
x_j^2 {\partial \over \partial x_j}  \Big )
\end{equation}
having the structure (\ref{4.0}). 

We will have use for the orthogonality of the Jack polynomials with respect to a particular scalar product, $\langle \! \langle f,g\rangle \! \rangle^{(\alpha)}$ say, which can be traced back to Jack's original paper on the subject \cite{Ja70}.
This scalar product is defined in terms of the power sum basis by
\begin{equation}\label{JS1}
\langle\langle p_\kappa ,p_\mu \rangle\rangle^{(\alpha)} =  \alpha^{\ell(\kappa)} z_\kappa  \delta_{\kappa, \mu}, \quad
z_\kappa:= \prod_{j=1}^{\kappa_1} j^{f_j(\kappa)} f_j(\kappa)!.
\end{equation}
Here $\ell(\kappa)$ is the number of non-zero parts in the partition $\kappa$, $f_j(\kappa)$ is the number of times a part of $\kappa$ is equal to $j$, while $p_\kappa =  \prod_{j=1}^{\ell(\kappa)}
p_{\kappa_j}$. Fundamental to the theory of Jack polynomials is that they provide a further basis of symmetric polynomials which is also orthogonal with respect to this scalar product \cite{Ja70}, \cite{Ma95},
\cite[Eq.~(12.130)]{Fo10}
\begin{equation}\label{JS2}
\langle \! \langle P_\kappa^{(\alpha)} ,P_\mu^{(\alpha)} \rangle \! \rangle^{(\alpha)} = {h_\kappa' \over h_\kappa}\delta_{\kappa, \mu}.
\end{equation}
Here 
\begin{equation}\label{2.2h}
h_\kappa' = \prod_{s \in \kappa} ( \alpha (a(s) + 1) + l(s)), \quad {\rm and} \quad h_\kappa = \prod_{s \in \kappa} (\alpha a(s) + l(s)+1),
\end{equation}
where the quantities $a(s), l(s)$ are specified in the text about (\ref{h1}) below.

\subsection{Preliminaries relating to the $(q,t)$-generalisation}\label{S1.3}

As is usual in the $q$-calculus, introduce
\begin{equation}\label{5.1}
(x;q)_a = {(x;q)_\infty \over (x q^a; q)_\infty}, \qquad(u;q)_\infty = \prod_{j=0}^\infty (1 - u q^j).
\end{equation}
Also introduce the particular $q$-generalisation of $e^{-x}$,
\begin{equation}\label{5.2}
E_q(-x) := \sum_{n=0}^\infty {(-1)^n q^{n(n-1)/2} x^n \over
(q;q)_n} = (x;q)_\infty;
\end{equation}
for the evaluation of the sum to the product form, which requires
$|q| < 1$, see \cite{GR11}. It follows from the sum form in (\ref{5.2}) that
\begin{equation}\label{5.2a}
\lim_{q \to 1} E_q(-x (1 - q)) = e^{-x},
\end{equation}
thus justifying the stated status of $E_q(-x)$ in relation to $e^{-x}$.
Using (\ref{5.2a}) together with the product form in (\ref{5.2}) shows
\begin{equation}\label{5.2b}
\lim_{q \to 1^-} E_q(-x (1 - q)^{1/2}) 
 E_q(x (1 - q)^{1/2}) = \lim_{q \to 1^-} E_{q^2}(-x^2 (1 - q)) =
 e^{-x^2/2}.
 \end{equation}
 Thus, up to scaling in $x$, $E_q(-x) E_q(x)$ can be chosen as a $q$-generalisation of the Gaussian weight $e^{-x^2}$. Up to normalisation, this then is the case $a=-1$ of the weight (\ref{5.2c}).

From the viewpoint of $q$-orthogonal polynomial theory,
the most natural $q$-generalisation of the eigenvalue PDF for the GUE is the functional form proportional to
\begin{equation}\label{KE1}
\prod_{l=1}^N w_U^{(-1)}(x_l;q) \prod_{1 \le i < j \le N }
(x_i - x_j)^2,
 \end{equation}
 supported on $\{ \pm 1, \pm q, \pm q^2,\dots \}$;
 see e.g.~\cite{FLSY23}. This differs from the $a=-1$,
 $m=1$ case of (\ref{5.3}) which instead is proportional to
 \begin{equation}\label{KE2}
\prod_{l=1}^N w_U^{(-1)}(x_l;q) \prod_{1 \le i < j \le N }
(x_i - x_j)(x_i - qx_j).
 \end{equation}
 However, by a lemma due to Kadell \cite[Lemma 4]{Ka88} (see also \cite[Th. 7.2]{KS17}),
 for any $\phi(\mathbf x)$ symmetric
  \begin{equation}\label{KE3}
  \langle  \phi(\mathbf x) \rangle^{(-1,q,q)}_N =
   \langle  \phi(\mathbf x) \rangle^{ \rm qGUE}.
 \end{equation}  
 Hence a quantity such as the spectral moments is the same in both cases. However, to $q$-generalise the Vandermonde product factor in (\ref{1.1a}) with $\beta = 2m$ in a way that preserves integrability, it has been long recognised in the theory of the Selberg integral that the deformation present in (\ref{5.3}) must be utilised; see \cite{As80}. 
While this deformation is not symmetric, so extending (\ref{KE1}) and applying the lemma of  Kadell allows (\ref{5.3}) to be modified to the symmetric functional form proportional to
 \begin{equation}\label{5.3x}
  \prod_{l=1}^N  w_U^{(a)}(x_l;q) \Big (\prod_{1 \le i < j \le N} (x_i -  x_j) \Big )
  \prod_{p=-(m-1)}^{m -1}\prod_{1 \le i < j \le N} (x_i - q^p x_j),
  \end{equation}
  and the analogue of (\ref{KE3}) remains true. With $t = q^m$ this is said to be the symmetric $(q,t)$ generalisation of (\ref{1.1a}), companion to (\ref{5.3}).

%   In particular, the case $a=-1$, $m=1$ of (\ref{5.3}) is said to define the qGUE, this being the most natural $q$-generalisation of the eigenvalue PDF for the GUE from the viewpoint of orthogonal polynomial theory; see e.g.~\cite{FLSY23}. The references given at the beginning of this subsection in relation to the computation of moments all relate to this case. A conjecture for the generalisation of the superintegrability evaluation formula (\ref{4.2}) 
%   in the case $\beta = 2$ has been formulated in \cite[Eqns.~(16) and (20) with $t=q$]{MPS18},
%    \cite[Eq.~(2.9)]{MPS20} and
%    \cite[Eq.~(25) with $\xi_1 = - \xi_2 = 1$ and $t=q$]{MM22}
%   To state this, the Schur polynomial is to be regarded as written in the power sum basis, which is a homogeneous linear combination involving products of $p_k := \sum_{l=1}^N x_l^k$, and a substitution is made for each $p_k$ by writing $\{ p_k = c_k\}$ for some specified $c_k$. With this notation, the cited references give
%   \begin{equation}\label{3.7q}
%  \langle  s_\kappa(\mathbf x) \rangle^{\rm qGUE} 
%  =
% s_\kappa\Big ( \Big \{ p_k = {1 - q^{kN} \over 1 - q^k}
% \Big \} \Big )
% {  s_\kappa \Big ( \Big \{ p_k = {(1 + (-1)^k) \over 1 - q^k}
% \Big \} \Big ) \over s_\kappa \Big ( \Big \{ p_k = {1  \over 1 - q^k}
% \Big \} \Big )}.
% \end{equation}
% This after scaling by $(1-q)^{|\kappa|/2}$, which corresponds to the scaling $x_l \mapsto (1 - q)^{1/2} x_l$
% in the arguments of the homogeneous polynomial $s_\kappa$ in keeping with the scaling of (\ref{5.2b})),
% is seen to reduce to (\ref{3.7}) in the limit $q \to 1$.

With the conjectured superintegrability identity 
(\ref{Mq}) in the case $q=t$ as the starting point (then the Macdonald polynomials reduce to the Schur polynomials), and with use made of (\ref{3.5}), a study of the moments
$m_{N,2p}^{\rm qGUE}$ was initiated in \cite{MPS20}. This leads to a particularly simple formula for the corresponding generating function (in $N$); see \cite[Eq.~(3-6)]{MPS20}. It was further developed in \cite[Cor.~(4.5)]{Co21}, resulting in a particular $q$-Hahn polynomial formula for the moments when $N$ is chosen randomly according to a certain negative binomial distribution.

To specify the Macdonald polynomials beyond the case $q=t$,
following \cite{Ma95} introduce the difference operator
\begin{equation}\label{D1}
D_N^1(q,t) := \sum_{1 \le j,k \le N} {t x_j - x_k \over x_j - x_k} \tau_j,
\end{equation}
where
\begin{equation}
(\tau_j f) (x_1,\dots,x_N) = f(x_1,\dots, qx_j, \dots, x_N).
\end{equation}
The Macdonald polynomials $P_\kappa(\mathbf x;q,t)$ are the symmetric polynomial eigenfunctions of (\ref{D1}) with the same structure 
(\ref{4.0}) as the Jack polynomials.  Moreover, they reduce to the Jack polynomials in an appropriate $q \to 1$ limit,
\begin{equation}\label{D2}
\lim_{q \to 1} P_\kappa(\mathbf x;q,q^{\beta/2}) =
P_\kappa^{(2/\beta)}(\mathbf x).
\end{equation}
Using (\ref{D2}) shows that (\ref{4.2}) can be reclaimed  as a limiting case of (\ref{Mq}) with $a=-1$. With regards to this, note that it is immediate that the first factor on the right hand side in (\ref{Mq}) limits to the corresponding factor in (\ref{4.2}). For the second factor on the right hand side of (\ref{Mq}), which is a ratio power sum specialised Macdonald polynomials, note that as $t \to 1$ and with $a=-1$ the leading terms comes from the coefficient of $p^{|\kappa|/2}$ in the numerator, and the coefficient of $p_1^{|\kappa|}$ in the denominator, with the factors in the coefficients coming from the substitutions, $1/(1-t)^{|\kappa|}$, cancelling out. This gives the second factor on the right hand side of (\ref{4.2}).
% Another noteworthy point is that for $q=t$ the Macdonald polynomials are in fact independent of $q$ and are given by the Schur polynomials. In terms of the Macdonald polynomials, the superintegrability conjectures of 
% \cite{MPS18},
%    \cite{MPS20} and \cite{MM22} (all after minor modification to match our setting) read
%  Upon scaling by $(1-q)^{|\kappa|/2}$, the limit $q \to 1$ of (\ref{Mq}) reclaims (\ref{4.2}) with $\beta = 2m$.

Our subsequent working will require some further quantities and results from
 Macdonald polynomial theory.
 One is the notion of the diagram of the partition of $\kappa$, where each part $\kappa_i$ is drawn as the left justified $i$-th row consisting of $\kappa_i$ unit squares. For each square $s = (i,j)$ in the diagram, one defines the arm and leg lengths according to $a(s) = \kappa_i - j$ and $l(s) = \kappa_j' - i$ respectively. 
In terms of this notation, we introduce
\begin{equation}\label{h1}
h_\kappa(q,t) := \prod_{s \in \kappa}(1 - q^{a(s)}
t^{l(s) + 1}).
\end{equation}
The conjugate partition $\kappa'$ is defined in terms of the diagram of $\kappa$ by reversing the role of each row and column. We set 
\begin{equation}\label{n1}
n(\kappa) := \sum_{i=1}^N (i - 1) \kappa_i := \sum_{i=1}^{\kappa_1} \kappa_i'(\kappa_i'-1)/2, 
\end{equation}
and define the generalised $(q,t)$ Pochhammer symbol 
\begin{equation}\label{gp}
(a)_\kappa^{(q,t)} := 
\prod_{(i,j) \in \kappa} (t^{i-1} - aq^{j-1}) =
t^{n(\kappa) }\prod_{i=1}^N (a t^{1-i}; q)_{\kappa_i}. 
\end{equation}
We will also require the $(q,t)$  analogue of the power sum scalar product (\ref{JS1}),
\begin{equation}\label{JS1q}
\langle \! \langle p_\kappa ,p_\mu \rangle \! \rangle^{(q,t)} =   z_\kappa \prod_{l=1}^{\ell(\kappa)}    {1 - q^{\kappa_l} \over 1 - t^{\kappa_l}}  \delta_{\kappa, \mu}.
\end{equation}
The significance of this scalar product in Macdonald polynomial theory is the additional orthogonality \cite{Ma95}
\begin{equation}\label{JS2q}
\langle \! \langle P_\kappa(\mathbf x;q,t) ,P_\mu({\mathbf x;q,t}  \rangle \! \rangle^{(q,t)} = {h_\kappa'(q,t) \over h_\kappa(q,t)}\delta_{\kappa, \mu}.
\end{equation}
Here the definition of $h_\kappa(q,t)$ is as given in 
(\ref{h1}), while
\begin{equation}\label{h1dx}
h_\kappa'(q,t) := \prod_{s \in \kappa}(1 - q^{a(s)+1}
t^{l(s) }).
\end{equation}

\section{Proofs of the superintegrability identities}
\subsection{The Gaussian $\beta$ ensemble superintegrability identity}\label{S2.1}
Although, as already remarked, several proofs of (\ref{4.2}) are available in the literature, it is instructive to revise the particular proof strategy given in \cite{De09} as preparation for our proof of the $(q,t)$ generalised  superintegrability identity (\ref{Mq}). 

% We begin by considering the parameter dependent Hermitian Gaussian ensembles as introduced in \S \ref{S1.1}. Let $d^{\rm H} U$ denote the normalised Haar measure on real orthogonal $(\beta = 1)$, complex unitary $(\beta = 2)$ and unitary symplectic matrices $(\beta = 4)$. We then have that the eigenvalue probability density function for these ensembles is, up to proportionality, given by \cite[Eq.~(11.101), with the change of notation $(U^\dagger dU) \mapsto d^{\rm H} U$]{Fo10}
% \begin{equation}\label{6.1g}
% \prod_{1 \le j < k \le N} | x_k - x_j|^\beta
% e^{- \tilde{\beta} \sum_{j=1}^N x_j^2 - \tilde{\beta} t^2 \sum_{j=1}^N
% (x_j^{(0)})^2}
% \int e^{2 \tilde{\beta} t {\rm Tr} ( UG U^\dagger X_0)} \,
% d^{\rm H} U,
% \end{equation} 
% where $\tilde{\beta} := \beta/2 |1 - e^{-2 \tau}|$, $t:=e^{-\tau}$ and $\mathbf x^{(0)}$ are the eigenvalues of 
% $X_0$.

Required for this purpose is a certain generalised hypergeometric function of two sets of variables based on Jack polynomials. 
With $h_\kappa'$ as in (\ref{2.2h}) this is defined as
\begin{equation}\label{6.3g}
{\vphantom{\mathcal F}}_0^{\mathstrut}\mathcal F_0^{(\alpha)}(\mathbf x; \mathbf y) = \sum_{\kappa} {\alpha^{| \kappa |} \over h_\kappa'} 
{P_\kappa^{(\alpha)}(\mathbf x) P_\kappa^{(\alpha)}(\mathbf y) \over
P_\kappa^{(\alpha)}(\mathbf y) |_{\mathbf y = \mathbf 1}
}.
\end{equation}
In this case $\beta = 1$ this special function was first introduced in the field of mathematical statistics (see the text \cite[Th.~7.3.3 with $p=q=0$]{Mu82}) 
via their integral form
\begin{equation}\label{6.2g}
{\vphantom{\mathcal F}}_0^{\mathstrut} \mathcal F_0^{(2/\beta)}(\mathbf x; \mathbf x^{(0)})
=
\int e^{ {\rm Tr} ( UX U^\dagger  X_0)} \,
d^{\rm H} U  .
\end{equation}
Here $d^{\rm H} U$ denotes the normalised Haar measure on the orthogonal group, and $\mathbf x$ ($\mathbf x_0$) are the eigenvalues of the real symmetric matrix $X$ ($X_0$). This integral form remains true for $\beta =2$ (with $U$ now from the unitary group, and $X,X_0$ complex Hermitian matrices) and also for $\beta = 4$ (with $U$ now from the symplectic unitary group, and $X,X_0$ quaternion Hermitian matrices). It arises in the study of the random matrix theory associated with (\ref{2.1a})
\cite[\S 11.6.2]{Fo10}.

% Moreover, the fact that (\ref{6.1g}) satisfies the Fokker-Planck equation (\ref{2.1a}) allows for a meaning to be given to ${\vphantom{\mathcal F}}_0^{\mathstrut}\mathcal F_0^{(2/\beta)}$ for general $\beta > 0$. In fact a result of \cite{BF97a} identifies this quantity as a certain generalised hypergeometric function of two sets of variables based on Jack polynomials. 
% With $h_\kappa'$ as in (\ref{2.2h}) we have
% \begin{equation}\label{6.3g}
% {\vphantom{\mathcal F}}_0^{\mathstrut}\mathcal F_0^{(\alpha)}(\mathbf x; \mathbf y) = \sum_{\kappa} {\alpha^{| \kappa |} \over h_\kappa'} 
% {P_\kappa^{(\alpha)}(\mathbf x) P_\kappa^{(\alpha)}(\mathbf y) \over
% P_\kappa^{(\alpha)}(\mathbf y) |_{y_1 = \cdots = y_N = 1}
% }.
% \end{equation}

Introduce now a measure corresponding to a scaled version of the Gaussian $\beta$ ensemble probability density (\ref{1.1a}),
\begin{equation}\label{1.1H}
d \mu^{\rm G}(\mathbf x) := {1 \over \tilde{Z}_N^{\rm G}}
\prod_{1 \le i < j \le N} | x_i - x_j |^{2/\alpha}
\prod_{l=1}^N e^{-x_l^2/2} dx_1 \cdots dx_N,
\end{equation}
where $\tilde{Z}_N^{\rm G}$ is the normalisation, and we have set $\beta = 2/\alpha$. A key property of ${\vphantom{\mathcal F}}_0^{\mathstrut}\mathcal F_0^{(\alpha)}(\mathbf x; \mathbf y)$ is its value when integrated against (\ref{1.1H})
\cite[special case $z_1=\cdots=z_N = 0$ of Proposition 3.8]{BF97a},
\begin{equation}\label{MH}
\int_{\mathbb R^N} \, {\vphantom{\mathcal F}}_0^{\mathstrut}\mathcal F_0^{(\alpha)}(\mathbf x; \mathbf y) \, d \mu^{\rm G}(\mathbf x) =
e^{\sum_{j=1}^N y_j^2/2}.
\end{equation}
(A proof for $\beta = 1,2$ and 4 based on matrix integration using (\ref{6.2g}) requires no more than completing the square, whereas the proof for general $\alpha >0$ in \cite{BF97a} requires first developing a theory of multidimensional Hermite polynomials based on Jack polynomials relating to (\ref{2.0}).)
In 
the proof of (\ref{4.2}) given in \cite{De09}, this identity is one of the two key ingredients (it is further the case that the original proof  of \cite{Ok97} makes use of an equivalent integral identity, and that (\ref{MH}) is also used in the proof of \cite{Du03}). The other key ingredient is the orthogonalities (\ref{JS1}) and (\ref{JS2}).

% For $2/\alpha = 1,2$ and 4, 
% and with $X_0$ replaced by $Y$ for notational convenience, the random matrix interpretation gives
% \begin{multline}\label{MH}
% \int_{\mathbb R^N} \, {\vphantom{\mathcal F}}_0^{\mathstrut}\mathcal F_0^{(\alpha)}(\mathbf x; \mathbf y) \, d \mu^{\rm G}(\mathbf x) = \Big \langle 
% e^{- {\rm Tr} \, G^2/2 } e^{ {\rm Tr} ( UG U^\dagger  Y)}
% \Big \rangle_{G, U} \\
% = e^{{\rm Tr} \, Y^2/2} 
% \Big \langle 
% e^{- {\rm Tr} \, (G - UG U^\dagger  Y)^2/2}
% \Big \rangle_{G, U} = e^{{\rm Tr} \, Y^2/2} = e^{\sum_{j=1}^N y_j^2/2}.
% \end{multline}
% In the averages over $G$, the normalisation is taken so that $\langle e^{- {\rm Tr} \, G^2/2 } \rangle = 1$. To obtain the first equality on the second line, a completion of the square has been carried out, while to obtain the second equality on the second line, the translation invariance of the space of Hermitian matrices with Lebesgue measure has been used.

% A natural question is to ask about the validity of (\ref{MH}) for general $\alpha > 0$. By developing a theory of multidimensional Hermite polynomials based on Jack polynomials, these being a complete set of symmetric orthogonal polynomials with respect to the measure (\ref{1.1H}), it was shown in \cite[special case $z_1=\cdots=z_N = 0$ of Proposition 3.8]{BF97a} that (\ref{MH}) is indeed valid for  general $\alpha > 0$. 

One proceeds by substituting (\ref{6.3g}) in the left hand side of (\ref{MH}) in terms of power sums according to
$$
e^{\sum_{j=1}^N y_j^2/2} = \sum_{k=0}^\infty {1 \over 2^k k!}  (
p_2(\mathbf y)  )^k.
$$
Equating polynomials homogeneous of degree $2k$ on both sides gives
\begin{equation} \label{2.2i}
\alpha^{2k} \sum_{\kappa: |\kappa| = 2k} {1\over h_\kappa'} 
 {P_\kappa^{(\alpha)}(\mathbf y) \over
P_\kappa^{(\alpha)}(\mathbf y) |_{\mathbf y = \mathbf 1}}
\int_{\mathbb R^N} {P_\kappa^{(\alpha)}(\mathbf x) \,
d \mu^{\rm G}(\mathbf x}) 
= {1 \over 2^k k!}  (
p_2(\mathbf y)  )^k.
\end{equation}
Next, on both sides of (\ref{2.2i}), take the inner product with respect to the
particular Jack polynomial $P_\kappa^{(\alpha)}(\mathbf y)$ with $|\kappa| = 2k$, making use of the the orthogonality (\ref{JS2}) on the left hand side. One reads off that
\begin{equation}\label{JS3}
 {\alpha^{2k} \over h_\kappa} 
 {1 \over
P_\kappa^{(\alpha)}(\mathbf y) |_{\mathbf y = \mathbf 1}}
\int_{\mathbb R^N} {P_\kappa^{(\alpha)}(\mathbf x) \,
d \mu^{\rm G}(\mathbf x}) 
= {1 \over 2^k k!}\langle \! \langle P_\kappa^{(\alpha)} , p_2^k \rangle \! \rangle^{(\alpha)}.
\end{equation}
Noting that for $\kappa = (2)^k$, $\alpha^{\ell(\kappa)} z_\kappa = \alpha^k 2^k k!$, consideration of the orthogonality (\ref{JS1}) on the right hand side of gives the equivalent form
\begin{equation}\label{JS4}
 \alpha^{k} 
 \int_{\mathbb R^N} {P_\kappa^{(\alpha)}(\mathbf x) \,
d \mu^{\rm G}(\mathbf x}) 
=  h_\kappa P_\kappa^{(\alpha)}(\mathbf y) |_{\mathbf y = \mathbf 1}
[ p_2^k] P_\kappa^{(\alpha)}(\mathbf y).
\end{equation}
This is the superintegrability identity (\ref{4.2}), upon knowledge of the fact that
\begin{equation}\label{2.2k}
[ p_1^{2k}] P_\kappa^{(\alpha)}(\mathbf y) = {1 \over h_\kappa}
\end{equation}
(see e.g.~\cite[Prop.~12.6.7]{Fo10} for an equivalent statement), and a simple change of variables on the left hand side.

\begin{remark}\label{R3x}
    As emphasised in \cite[\S 8.4]{Du03} and \cite[\S 5.3]{MRW15}, there is an alternative evaluation of 
    (\ref{4.2})
    \cite[equivalent to the case $\mathbf z = \mathbf 0$ of Corollary 3.2]{BF97a},
    \begin{equation}\label{3.9a}
        (2/\beta)^{|\kappa|/2}
        \Big \langle P_\kappa^{(2/\beta)}(\mathbf x) \Big \rangle^{{\rm G}\beta{\rm E}_N} = P_\kappa^{\rm H}(\mathbf 0;2/\beta).
    \end{equation}
    Here $P_\kappa^{\rm H}(\mathbf z;\alpha)$ denotes the generalised Hermite polynomials based on Jack polynomials; see \cite[\S 13.3]{Fo10}. There is a recursive formula for the computation of the coefficients $\{ a_{\kappa,\mu} \}$ in the expansion
    $$
    P_\kappa^{\rm H}(\mathbf z;\alpha) = 
    P_\kappa^{(\alpha)}(\mathbf z) +
    \sum_{\mu \subset\kappa}
    a_{\kappa,\mu}
    P_\mu^{(\alpha)}(\mathbf z)
    $$
    (the notation $\mu \subset\kappa$ indicates that the diagram of $\mu$ must be strictly contained in the diagram of $\kappa$),
    which moreover has been implemented in software \cite{DES07}. In light of the explicit formula (\ref{4.3}), and noting that $P_\kappa^{\rm H}(\mathbf 0;\alpha)=a_{\kappa,\mathbf 0}$, this can then be viewed as a computation scheme for the evaluation of $[p_2^{|\kappa/2|}] P_\kappa^{(2/\beta)}(\mathbf x)$.
\end{remark}

\subsection{The $(q,t)$ generalised Gaussian ensemble superintegrability identity --- proof of Theorem \ref{T1}}\label{S3.2s}
Here we will show that the proof of (\ref{4.2}) from \cite{De09}, as revised in the previous subsection, can be generalised to provide a proof of (\ref{Mq}). Following \cite[Eq.~(1.1)]{BF98} as a $(q,t)$-analogue of the generalised hypergeometric function of two sets of variables (\ref{6.3g}) we introduce
\begin{equation}\label{8.1}
{}_0 \mathcal F_0(\mathbf x, \mathbf y;q,t) :=
\sum_\kappa {t^{n(\kappa)} \over h_\kappa'(q,t) 
P_\kappa(1,t,\dots,t^{N-1};q,t)} P_\kappa(\mathbf x;q,t)
P_\kappa(\mathbf y;q,t).
\end{equation}
Here the quantity $n(\kappa)$ is as defined below (\ref{h1}), while $h_\kappa'(q,t)$ is given by (\ref{h1dx}).
One notes too that
\begin{equation}\label{8.2}
P_\kappa(1,t,\dots,t^{N-1};q,t) =
P_\kappa\Big ( \Big \{ p_k = {1 - t^{kN} \over 1 - t^k}
\Big \}; q,t \Big ).
\end{equation}

It has already been remarked that the Al-Salam and Carlitz polynomials are the $q$-orthogonal polynomials corresponding to the weight (\ref{5.2c}). The theme of the work \cite{BF98} was to develop a theory of multivariable Al-Salam and Carlitz $q$-orthogonal polynomials corresponding to the multivariable weight (\ref{5.3}). In the course of this study, an integration formula providing a $(q,t)$ generalisation of (\ref{MH}) was obtained. This reads
\cite[Prop.~4.8 with $\mathbf z = \mathbf 0$]{BF98}
\begin{equation}\label{8.3}
 \Big \langle {}_0 \mathcal F_0(\mathbf x, \mathbf y;q,t) \Big \rangle^{(a,q,t)}_N =
\prod_{l=1}^N {1 \over E_q(-y_l) E_q(-ay_l)}.
\end{equation}

We know from the proof of (\ref{4.2}) revised in the previous subsection that the identity (\ref{MH}) is one of two key ingredients in its proof. This is seen to be appropriately generalised by way of (\ref{8.3}). The
second key ingredient, generalising from the Jack to the Macdonald case, relates to the scalar product (\ref{JS1q}).
To make use of this scalar product, we require the power sum expansion of the right hand side of (\ref{8.3}).
For this the product form of (\ref{5.2}) is substituted, and the exponential of the logarithm is taken. Expanding the logarithm as a power series shows
\begin{multline*}
- \log \prod_{l=1}^N E_q(-y_l) E_q(-ay_l)  =
\sum_{l=1}^N \sum_{j=0}^\infty \sum_{r=1}^\infty \bigg (
{y_l^r q^{jr} \over r} + {y_l^r a^r q^{jr} \over r} \bigg )
\\
= \sum_{j=0}^\infty \sum_{r=1}^\infty \bigg (
{p_r q^{jr} \over r} + {p_r a^r q^{jr} \over r} \bigg ) 
 = \sum_{r=1}^\infty \bigg (
{p_r  \over (1 - q^r) r} + {p_r a^r  \over (1 -  q^r) r} \bigg ),
\end{multline*}
where $p_r$ denotes the $r$-th power sum in the variables $y_1,\dots,y_N$.
Hence
\begin{equation}\label{8.4}
\prod_{l=1}^N {1 \over E_q(-y_l) E_q(-ay_l)} =
\prod_{r=1}^\infty \exp \bigg ( {(1 + a^r) p_r  \over (1 - q^r) r} \bigg ).
\end{equation}

We substitute (\ref{8.1}) in the left hand side of (\ref{8.3}), and (\ref{8.4}) on the right hand side. This allows for the terms that are homogeneous polynomials in $\{y_j\}$ of degree $k$ on the left hand side, and on the right hand side, to be read off. Hence, as a $(q,t)$ generalisation of (\ref{2.2i}) we deduce
\begin{multline}\label{2.2iq}
 \sum_{\kappa: |\kappa| = k} {t^{n(\kappa)} \over h_\kappa'(q,t) P_\kappa(1,t,\dots,t^{N-1};q,t)} 
\Big \langle P_\kappa(\mathbf x;q,t) \Big \rangle^{(a,q,t)}_N P_\kappa(\mathbf y;q,t) \\
= \sum_{\substack{f_1,\dots,f_k \ge 0 \\ \sum_{j=1}^{\kappa_1} j f_j = k}} \prod_{r=1}^k {(1 + a^r)^{f_r} \over r^{f_r} f_r! (1 - q^r)^{f_r}} p_{\kappa_{\{ f_j \}}}(\mathbf y).
\end{multline}
Here, on the right hand side, $f_j$ have the same meaning as in (\ref{JS1}), and the notation
$\kappa_{\{ f_j \}}$ indicates that the partition is to be thought of as defined by these variables. From here we use the scalar product orthogonality (\ref{JS2q}) 
to read off from this that for a particular partition $\kappa$ with $| \kappa | = k$
\begin{equation}\label{2.2jq}
 {t^{n(\kappa)} \over  h_\kappa(q,t)} 
 {\Big \langle P_\kappa(\mathbf x;q,t) \Big \rangle^{(a,q,t)}_N \over 
 P_\kappa(1,t,\dots,t^{N-1};q,t) }
= \sum_{\substack{f_1,\dots,f_k \ge 0 \\ \sum_{j=1}^{\kappa_1} j f_j = k}} \prod_{r=1}^k {(1 + a^r)^{f_r} \over r^{f_r} f_r! (1 - q^r)^{f_r}} \Big \langle \! \Big \langle
 P_\kappa(\mathbf y;q,t),p_{\kappa_{\{ f_j \}}}(\mathbf y)
\Big \rangle \! \Big \rangle^{(q,t)} .
\end{equation}
Use of (\ref{JS1q}) in this expression identifies the right hand side as equal to
\begin{equation}\label{2.24}
P_\kappa \Big ( \Big \{ p_k = {1 + a^k \over 1 - t^k}
\Big \};q,t \Big ).
\end{equation}
Also, it follows from the general property of
Macdonald polynomials
 \cite[Eq.~(5.3)]{Ma95}
\begin{equation}\label{Me}
P_\kappa\Big ( \Big \{ p_k = {1 - u^k \over 1 - t^k}
\Big \}; q,t \Big ) = {(u)_\kappa^{(q,t)} \over h_\kappa(q,t)},
\end{equation}
valid for general $u$,
that
\begin{equation}\label{3.17}
P_\kappa \Big ( \Big \{ p_k = {1 \over 1 - t^k}
\Big \}; q,t \Big ) =
 {t^{n(\kappa)} \over  h_\kappa(q,t)}. 
\end{equation}
Use of this together with (\ref{8.2}) shows that (\ref{2.2jq}) with the substitution of (\ref{2.24}) for the right hand side is precisely the superintegrability identity (\ref{Mq}). \hfill $\square$

\begin{remark}\label{R3.1}\label{RE1} $ $
\\ 1.
One observes that the requirement $t = q^m$ for $m$ a positive integer, as assumed in the PDF (\ref{5.3}) and so allowing for the definition of $\langle \cdot \rangle^{(a,q,t)}_N$, plays no role in the above proof. How to define  $\langle \cdot \rangle^{(a,q,t)}_N$ as a multidimensional Jackson integral beyond the case $t = q^m$
is known; see
\cite{St97,St00} and \cite[\S 5]{BF98}. But while from (\ref{5.3}) $\langle \cdot \rangle^{(a,q,t)}_N$ 
is a single multidimensional Jackson integral,
the more general case requires introducing a sum of $N+1$ weighted multidimensional Jackson integrals; see \cite[Eq.~(5.2) with $a=b=0$, $c=1$ and $d=-a$]{St97}.
More on this class of weighted multidimensional Jackson integrals can be found in \cite{IF14,IF17}. \\
2. In the case $m=1$ and thus $q=t$, when the Macdonald polynomials reduce to the Schur polynomials, the average in (\ref{Mq}) for $a=-1$ has been evaluated in an explicit product form in \cite[Prop.~2.3]{FLSY23}. This generalises the known explicit product form  for $\langle s_\kappa(\mathbf x) \rangle^{\rm GUE}$ \cite{DI93}.
\end{remark}

\section{Application to dualities and moments}\label{S4}
The results of Corollaries \ref{C1q} and  \ref{P5} both have analogues for the Gaussian $\beta$ ensemble. In our presentation these results will be stated and proved.
As with the material of \S \ref{S2.1}, although the results and their proofs are already available in the literature,
knowledge of their derivation is necessary to guide the pathway to their $(q,t)$ analogues.
Moreover, where needed, we give modifications and refinements of the original proofs to best allow for this.

\subsection{Duality formula for the G$\beta$E average with respect to Jack polynomials}\label{S3.1}

In the Jack case, a consequence of the superintegrability identity (\ref{4.2}) has been shown to be a duality formula for the G$\beta$E average of a Jack polynomial.

\begin{proposition}\label{C1}
 Let  $\alpha = 2/\beta$, and denote by $\ell(\kappa)$ the number of nonzero parts in $\kappa$.
For $\ell(\kappa) \le N$ and $\ell(\kappa') \le M$ we have
\begin{equation}\label{GF}
\Bigg \langle {P_\kappa^{(\alpha)}(\mathbf x) \over 
P_\kappa^{(\alpha)}(\mathbf x) |_{\mathbf x = \mathbf 1}} 
\Bigg \rangle^{{\rm G}\beta {\rm E}_N} 
=
(-\alpha)^{|\kappa|/2}
\Bigg \langle {P_{\kappa'}^{(1/\alpha)}(\mathbf x) \over 
P_{\kappa'}^{(1/\alpha)}(\mathbf x) |_{\mathbf x = \mathbf 1}} 
\Bigg \rangle^{{\rm G}(4/\beta) {\rm E}_M}. 
\end{equation}  
(Note in particular that under the stated conditions there is no dependence on $N$ or $M$ in this identity.)
\end{proposition}

This was first derived by Dumitriu \cite[Th.~8.5.3]{Du03} (and further developed by Dumitriu and Edelman in \cite[Lemma 2.6]{DE05}), and independently by Desrosiers \cite[Prop.~4]{De09}.
As the proofs of \cite{Du03} and \cite{De09} use (\ref{4.2}) with the substitution (\ref{4.3}), it turns out that they have extra complexity relative to using (\ref{4.2}) directly,
as can be seen by comparing with the derivation we will now give.

{\it Proof of Proposition \ref{C1}}. 
The essential property of Jack polynomials that must be used in conjunction with (\ref{4.2}) is 
 an automorphism on power sums $\omega_c p_\kappa = c^{\ell(\kappa)} p_\kappa$, present in the work of  Stanley \cite{St89} and Macdonald \cite{Ma95}.
 According to these works, for the choice $c= - \alpha$ this has the action on Jack polynomials
\begin{equation}\label{SM}
\omega_{-\alpha} P_\kappa^{(\alpha)}(\mathbf x) = (-1)^{|\kappa|}{h_\kappa' \over h_\kappa} 
 P_{\kappa'}^{(1/\alpha)}(\mathbf x).
\end{equation}
In the statement of \cite{St89} and \cite{Ma95} of this result the Jack polynomials are normalised differently. 
For (\ref{SM}) in the present normalisation see \cite[Prop.~12.8.2]{Fo10}. Actually, for our application the details of the prefactors on the RHS of (\ref{SM}) play no role. All that matters is the proportionality, which implies the result
\begin{equation}\label{SM1}
{ [p_2^{|\kappa|/2}] P_\kappa^{(\alpha)}(\mathbf x) \over [p_1^{|\kappa|}] P_\kappa^{(\alpha)}(\mathbf x )} =
(- \alpha)^{|\kappa|/2} 
{ [p_2^{|\kappa|/2}] P_{\kappa'}^{(1/\alpha)}(\mathbf x) \over [p_1^{|\kappa|}] P_{\kappa'}^{(1/\alpha)}(\mathbf x )}.
\end{equation}
We note that both sides are independent of the number of arguments in $\mathbf x$ (up to the proviso that this number is not less than $\ell(\kappa)$ on the left hand side, and not less than $\ell(\kappa')= \kappa_1$ on the right hand side).
Thus, substituting for each side  the appropriate Gaussian $\beta$ ensemble average as implied by (\ref{4.2}), the duality (\ref{GF}) results.
\hfill $\square$

\subsection{The $(q,t)$ generalised duality formula --- proof of Corollary \ref{C1q} and an application}\label{S3.1b}
Following the strategy of the proof of Proposition \ref{C1} given above, we first introduce the $(q,t)$ generalisation of the power sum automorphism $\omega_{-\alpha}$. This is the so-called Macdonald automorphism \cite{Ma95}
\begin{equation}\label{MA}
\omega_{q,t} p_\kappa = (-1)^{|\kappa| - \ell(\kappa)}
\bigg ( \prod_{j=1}^{\ell(\kappa)} {1 - q^j \over 1 - t^j}
\bigg ) p_\kappa. 
\end{equation}
Its action on the Macdonald polynomials is given by \cite{Ma95}
\begin{equation}\label{MA1}
\omega_{q,t} P_\kappa(\mathbf x;q,t) : =
(-1)^{|\kappa| }
P_\kappa \bigg (
\bigg \{ p_k \mapsto  - {1 - q^k \over 1 - t^k} p_k \bigg \};q,t \bigg )
= {h_{\kappa}'(q,t)\over h_{\kappa}(q,t)}  P_{\kappa'}(\mathbf \{ p_k \};t,q).
\end{equation}

By substituting $p_k = -(1+a^k)/(1-q^k)$, then separately
$p_k = -1/(1-q^k)$, it follows from this that
\begin{equation}\label{MA3}
{  P_\kappa \Big ( \Big \{ p_k = {1 + a^k  \over 1 - t^k}
\Big \};q,t \Big ) \over P_\kappa \Big ( \Big \{ p_k = {1 \over 1 - t^k}
\Big \}; q,t \Big )} =
{  P_{\kappa'} \Big ( \Big \{ p_k = -{1 + a^k  \over 1 - q^k}
\Big \};t,q \Big ) \over P_{\kappa'} \Big ( \Big \{ p_k = -{1 \over 1 - q^k}
\Big \}; t,q \Big )} = 
{  P_{\kappa'} \Big ( \Big \{ p_k = {1 + a^k  \over 1 - q^{-k}}
\Big \};t^{-1},q^{-1} \Big ) \over P_{\kappa'} \Big ( \Big \{ p_k = {1 \over 1 - q^{-k}}
\Big \}; t^{-1},q^{-1} \Big )}.
\end{equation}
To obtain the second equality, the general facts that 
\begin{equation}\label{4.6a}
P_\mu(\mathbf x;q,t) = P_\mu(\mathbf x;q^{-1},t^{-1}), \qquad P_\mu(\{c^k p_k\};q,t) = c^{|\mu|} P_\mu(\{ p_k\};q,t)
\end{equation}
have been used.

To make use of (\ref{MA3}), divide both sides of (\ref{Mq}) by the first factor on its right hand side. 
We then recognise the new right hand side as the left hand side of (\ref{MA3}). Now do the same in the case of (\ref{Mq}) with $\kappa$ replaced by $\kappa'$, $(q,t) \mapsto (t^{-1}, q^{-1})$, and the number of variables on the left hand side changed from $N$ to $M$, with $M$ greater than or equal to the number of parts in $\kappa'$. The new right hand side is the equal to the right hand side in (\ref{MA3}).
Making use too of (\ref{8.2}) 
we arrive at (\ref{GFq}). 
\hfill $\square$

An application of Corollary \ref{C1q} is to evaluate
$\Big \langle \prod_{l=1}^N (u - x_l) \Big\rangle^{(a,q,t)}_N$, which in a random matrix context corresponds to the averaged characteristic polynomial.

\begin{proposition}
We have
  \begin{equation}\label{4.6c}  
\bigg \langle \prod_{l=1}^N (u - x_l) \bigg \rangle^{(a,q,t)}_N = U_N^{(a)}(u;t), 
\end{equation}
where ${U}_N^{(a)}(u;q)$ denotes the $N$-th (monic) Al-Salam and Carlitz orthogonal polynomial as corresponds to the weight (\ref{5.2c}).
\end{proposition}

\begin{proof}
To begin, we consider (\ref{GFq}) with $M=1$, which requires that $\kappa = (1)^k$ (i.e.~the part 1 repeated $k$ times). After cross multiplying the denominators, and multiplying through by $(-1/u)^k$, the sums on both sides can be evaluated using the dual Cauchy identity (see e.g.~\cite[Th.~6.5]{No23}).
Minor manipulation then shows
\begin{equation}\label{dd1}
\bigg \langle \prod_{l=1}^N (u - x_l) \bigg \rangle^{(a,q,t)}_N =
u^N 
\langle 
(x/u;t)_N
\rangle^{(a,q,t)}_1
 \Big |_{q \mapsto t^{-1}} =: I_N(u;t).
\end{equation}

Our task now is to evaluate the one-dimensional $q$-integral $I_N(u;t)$.
Writing the factor of the integrand $(x/u;t)_N$ as a series in $(x/u)$
according to the $q$-binomial expansion shows
\begin{align*}
I_N(u;t) & =  \sum_{k=0}^N (-1)^k t^{k(k-1)/2}u^{N-k}
{N \brack k}_t  
\langle 
x^k
\rangle^{(a,q,t)}_1
 \Big |_{q \mapsto t^{-1}}
 \\
& = \sum_{k=0}^N (-1)^k t^{k(k-1)/2}u^{N-k}
{N \brack k}_t \sum_{l=0}^k {k \brack l}_{t^{-1}} a^l, 
\end{align*}
where here the second line follows from knowledge of the moments of the  weight $ w_U^{(a)}(x;q)$ \cite[Eq.~(5.115)]{NS13}.
To simplify further, we extract from this the coefficient of $a^{N-l}$, to be denoted $[a^{N-l}] I_N(u;t)$, which requires replacing $\sum_{l=0}^k {k \brack l}_{t^{-1}} a^l$ by ${k \brack N-l}_{t^{-1}}$, and replacing $\sum_{k=0}^N$ by $\sum_{k=N-l}^N$. Now shifting $k \mapsto N - l + s$ ($s=0,\dots,l$) shows
\begin{align*}
[a^{N-l}] I_N(u;t) & = (-1)^{N-l} u^l \sum_{s=0}^l
t^{(N-l+s)(N-l+s-1)/2} (-u)^{-s}
{N \brack l-s}_t {N - l + s \brack s}_{t^{-1}} \\
 & = t^{N (N-1)/2} (-1)^{N} (t u)^l 
 {(t^{-N};t)_l \over (t;t)_l} \sum_{s=0}^l (-u)^{-s} t^{s(s-1)/2} {k \brack l}_t \\
  & = t^{N (N-1)/2} (-1)^{N} (tu)^l 
 {(t^{-N};t)_l \over (t;t)_l} (1/u;t)_l,
\end{align*}
where the second line follows by manipulating the products implied by the definition of the $q$-binomial coefficients, and the third line follows from the $q$-binomial summation. From the explicit formula for 
$U_N^{(a)}(u;q)$ \cite[Eq.~(14.24.1)]{KLS10}, this tells us that $ I_N(u;t) = U_N^{(a)}(u;t)$, which when substituted in (\ref{dd1}) gives the sought result
(\ref{4.6c}).
\end{proof}

\begin{remark}\label{R4.3} $ $ \\
1.~In the case $q=t$ (i.e.~the case $m=1$ of (\ref{5.3})), the equivalence between (\ref{KE2}) and (\ref{KE1}) (which holds independent of setting $a=-1$) together with the classical Heine identity from orthogonal polynomial/ random matrix theory \cite[Eq.~(5.2.9)]{EKR15} give that
\begin{equation}\label{U1}
\bigg \langle \prod_{l=1}^N (u - x_l) \bigg \rangle^{(a,q,q)}_N = {U}_N^{(a)}(u;q),
\end{equation}
as is consistent with (\ref{4.6c}).
\\
2.~Under fairly general conditions,
for a random matrix ensemble with a well defined global density $\rho^{\rm g}(x) \mathbbm 1_{x \in I}$, the large $N$ asymptotic formula
$$
\bigg \langle \prod_{l=1}^N (c_N u - x_l) \bigg \rangle \sim \exp\bigg (-N \int_I 
\log (u - x) \, \rho^{\rm g}(x) \, dx \bigg  ),
$$
is expected to hold. Here $c_N$ is the scaling factor corresponding to the global scaling (e.g.~for the GUE, $c_N = \sqrt{N}$ and $\rho^{\rm g}(x)$ is the Wigner semi-circle law (\ref{3.3})). Moreover, in cases that the averaged characteristic polynomial satisfies a differential equation,
it has been shown in \cite[\S 3]{FL15} how the asymptotic formula then leads to 
 a characterising equation for the resolvent
\begin{equation}\label{4.10a}
G(u) := \int_I {\rho^{\rm g}(x) \over u - x} \, dx.
\end{equation}
In \cite[Remark 3.7.2]{Fo22}, this strategy was successfully implemented in the case of the Stietljes-Wigert random matrix ensemble, for which the characteristic polynomial satisfies a $q$-difference equation. The latter circumstance is in common with (\ref{U1}). Taking $a=-1$, when the Al-Salam and Carlitz polynomial $U_N^{(a)}(u;q)$ reduces to the discrete $q$-Hermite polynomial, use of the $q$-difference equation satisfied by the latter \cite[Eq.~(14.28.5)]{KLS10}, 
together with the asymptotic formula (\ref{4.10a}) with $c_N=1$, and setting too $q=e^{-\lambda/N}$, one is led to the limiting equation
$$
\Big (2 \sinh {G(u) \over 2} \Big )^2 =
(1 - e^{-\lambda}) u^2.
$$
However, in contrast to the case of the  Stietljes-Wigert random matrix ensemble, this is not consistent with the known exact evaluation of $G(u)$ in this setting
obtained recently in \cite[Eq.~(4.29)]{BFO24}; one might speculate that the domain no longer being continuous but rather a $q$-lattice is a cause of this.
\\
3.~Not restricting to $M=1$ in 
(\ref{GFq}) and proceeding as in the derivation of (\ref{dd1}) gives the duality
  \begin{equation}\label{Z1}  
\bigg \langle \prod_{m=1}^M \prod_{l=1}^N (u - q^{-(m-1)}x_l) \bigg \rangle^{(a,q,t)}_N =
\bigg \langle \prod_{j=1}^N \prod_{k=1}^M (u - t^{j-1}x_k) \bigg \rangle^{(a,t^{-1},q^{-1})}_M.
\end{equation}
The analogue of this duality for the Gaussian $\beta$ ensemble \cite[Eq.~(13.162)]{Fo10}, in the case of $\beta$ even, makes possible asymptotic analysis of the corresponding eigenvalue density \cite{DF06a,FT19a}. However, in the present $(q,t)$ setting one can no longer identify the left hand side of (\ref{Z1}) in terms of the density, which itself should be considered with respect to the symmetrised PDF (\ref{5.3x}). Leaving this point aside, there is also the difficulty of the interpretation of the right hand side as a Jackson integral   (recall Remark \ref{R3.1}.1), which would represent an added complexity to possible asymptotic analysis.
\end{remark}

\subsection{G$\beta$E moments}\label{S4.3}
Extending the notation $\langle \cdot \rangle^{\rm GUE}$ used above (\ref{3.1}), write
\begin{equation}\label{mb}
m_{N,2p}^{{\rm G}\beta{\rm E}}:=\Big \langle \sum_{l=1}^N x_l^{2p} \Big \rangle^{{\rm G}\beta {\rm E}_N}.
\end{equation}
We would like to make use of the superintegrability formula (\ref{4.2}) for the computation of these moments. For this, knowledge of the Jack polynomial generalisation of (\ref{3.5})  is required. In terms of the quantity $h_\kappa'$ defined in (\ref{2.2h}), and with $(u)_j := u(u+1) \cdots (u+j-1)$ denoting the increasing Pochhammer symbol, this is known to be given by (see e.g.~\cite[Eq.~(12.145)]{Fo10})
\begin{equation}\label{3.5J}
\sum_{l=1}^N x_l^j =  \sum_{\kappa: |\kappa| = j} 
u_\kappa(\alpha)
 P_\kappa^{(\alpha)}(\mathbf x), \quad u_\kappa(\alpha) :=
|\kappa| \alpha^{|\kappa|} {(\kappa_1 - 1)! \over h_\kappa'} \prod_{s = 2}^{\ell(\kappa)} (-(s-1)/\alpha)_{\kappa_s}.
\end{equation}
Substituting in (\ref{mb}) then applying (\ref{4.2}) gives an explicit formula for $m_{2k}^{{\rm G}\beta{\rm E}}$, up to the need to perform a sum over all partitions of size $2k$, and the need to compute 
$[p_2^{|\kappa|/2}] P_\kappa^{(2/\beta)}(\{p_\kappa \})$ (for a listing of the Jack polynomials expanded in the power sum basis for partitions up to and including length three, see \cite[Eq.(2.42) with $\beta \mapsto \beta/2$]{Lo21}). These requirements  make the resulting formula impractical to implement as a hand calculation apart from small orders,
\begin{align}\label{3.6}
2 \, m_{N,2}^{{\rm G}\beta{\rm E}} & = N^2 + N(-1 + \alpha) \nonumber \\
2^{2} \, m_{N,4}^{{\rm G}\beta{\rm E}} & = 2N^3 + 5N^2(-1 + \alpha)+N (3 - 5 \alpha + 3 \alpha^2),
\end{align}
where $\alpha = 2/\beta$.
Higher order tabulations of these moments are known in the literature, derived from the so-called loop equation method \cite{BMS11,MMPS12,WF14}, and also the use of (\ref{3.9a}) \cite{WF14}, which together provide the explicit form of $m_{2k}$ up to and including $k=10$. More recently, La Croix \cite[\S 1.1]{LC20} has provided Maple code to implement a recursion from his thesis  \cite[Lemma 4.13]{LC09} (a statement of this recursion can also be found in \cite[Eq.~(1.9)]{Ag18}), which is demonstrated to be far more efficient still.

In another direction, the expression for the moments obtained from this formalism can be used to deduce some fundamental structural properties, first isolated by Dumitriu and Edelman \cite[Th.~2.8]{DE05}. These authors made use of a combination of theory relating to Jack polynomials, theory relating to  multidimensional Hermite polynomials orthogonal with respect to (\ref{1.1H}),  a tridiagonal matrix realisation of (\ref{1.1a}) as an eigenvalue probability density function \cite{DE02} as well as the duality (\ref{GF}).

The first of the structural findings of \cite{DE05} is the generalisation of (\ref{3.1}),
\begin{equation}\label{3.1g}
2^p m_{N,2p}^{{\rm G}\beta{\rm E}} =
\sum_{g=0}^{p}
\mathcal E_g(p,\alpha) N^{p+1-g},
\end{equation}
where $\mathcal E_g(p,\alpha)$ is a polynomial in $\alpha$ of degree at most $p$ and with integer coefficients. From the formula for $m_{N,2p}^{{\rm G}\beta{\rm E}}$ obtained from (\ref{3.5J}) and (\ref{4.2}) one sees that
the only term depending on $N$ is that in (\ref{4.3}), which by inspection is of degree $|\kappa| = 2p$. The cancellations which are responsible for the refined expansion in $N$ (\ref{3.1g}) cannot be anticipated from this form (the same would seem to apply for the formalism of 
\cite{DE05}). The argument in \cite{DE05} for  $\mathcal E_g(p,\alpha)$ being a polynomial in $\alpha$ of degree at most $p$ and with integer coefficients is based on the tridiagonal matrix realisation of (\ref{1.1a}). In the present formalism this property is not immediate as one is faced with a rational function in $\alpha$, since the factor $h_\kappa'$ in the denominator in (\ref{3.5J}) is itself a polynomial in $\alpha$.

The second structural finding of \cite{DE05} in relation to (\ref{3.5J}) is the functional equation 
\begin{equation}\label{fa}
m_{N,2p}^{{\rm G}\beta{\rm E}} = (-2/\beta)^{p+1} 
m_{-\beta N/2,2p}^{{\rm G}(4/\beta){\rm E}}.
\end{equation}
This is illustrated by the explicit results (\ref{3.6}), upon recognising
that replacing $\beta$ by $4/\beta$ is equivalent to replacing $\alpha$ by $1/\alpha$. 
Since the left hand side is a polynomial in $N$, it is well defined beyond positive integer values, which then gives meaning to the right hand side.
The analogue of (\ref{fa}) for the moments (\ref{mbq}) can be obtained.
It is instructive to first revise the proof of (\ref{fa}), which is based on (\ref{3.5J}) and 
the duality (\ref{GF}).

{\it Proof of (\ref{fa}). } Our working is a refinement of that given in \cite{DE05}. 
We begin by noting that the product over $s$ in the definition of $u_\kappa(\alpha)$ in (\ref{3.5J}) multiplied by the factor $(\kappa_1-1)!$, to be denoted $r_\kappa(\alpha)$ say, can be written in terms of the diagram of $\kappa$ according to
$$
r_\kappa(\alpha) = \prod_{s \in \kappa \backslash (1,1)}(-(\ell(s)-1)/\alpha + (a(s) - 1)).
$$
 After taking the transpose of the diagram, simple manipulation gives
$r_\kappa(\alpha) = (-1/\alpha)^{|\kappa|-1} 
r_{\kappa'}(1/\alpha)$. 
It is also true that $h_\kappa'(\alpha) = \alpha^{|\kappa|} h_{\kappa'}(1/\alpha)$, as follows from (\ref{2.2h}). These results together in the definition of $u_\kappa(\alpha)$ in (\ref{3.5J}) show
\begin{equation}\label{Pt2}
(- \alpha)^{|\kappa| - 1} {u_\kappa(\alpha) \over h_\kappa(\alpha)} = {u_{\kappa'}(1/\alpha) \over h_{\kappa'}(1/\alpha)}.
\end{equation}

To utilise (\ref{Pt2}) in the derivation of (\ref{fa}), rewrite  the first equation in (\ref{3.5J}) as
\begin{equation}\label{3.5L}
\sum_{l=1}^N x_l^j = \alpha^{j} \sum_{\kappa: |\kappa| = j} 
{ u_\kappa(\alpha) \over h_\kappa(\alpha)} [N/\alpha]_\kappa^{(\alpha)} \bigg ( 
 {P_\kappa^{(\alpha)}(\mathbf x) \over
 P_\kappa^{(\alpha)}(\mathbf x)|_{\mathbf x = \mathbf 1}} \bigg ).
 \end{equation}
 In obtaining this form, use has been made of (\ref{4.3}) and (\ref{2.2k}).
 Averaging over $\mathbf x$ with the G$\beta$E${}_N$ measure allows for the substitution (\ref{GF}) with  $M=N$. Substituting too according to (\ref{Pt2}), and using the identity $[N/\alpha]_\kappa^{(\alpha)} = (-\alpha)^{-|\kappa|} [-N]_{\kappa'}^{(1/\alpha)}$ (see \cite[Exercises 12.4, q.2]{Fo10}) now has each term in the summand a function of $\kappa'$. Renaming $\kappa'$ as $\kappa$ then shows
 \begin{equation}\label{5.5}
 m_{N,2p}^{{\rm G}\beta{\rm E}} = (-2/\beta)^{p+1} 
 \alpha^{-2p} \sum_{\kappa: |\kappa| = 2p} 
 { u_\kappa(1/\alpha) \over h_\kappa(1/\alpha)} [-N]_\kappa^{(1/\alpha)} \bigg \langle 
 {P_\kappa^{(1/\alpha)}(\mathbf x) \over
 P_\kappa^{(1/\alpha)}(\mathbf x)|_{\mathbf x = \mathbf 1}} \bigg \rangle^{{\rm G}(4/\beta){\rm E}_N}.
 \end{equation}
 Important here is the fact, as follows from the proof of Proposition \ref{C1}, that the average on the right hand side is independent of $N$. Using this latter point,
 we recognise the right hand side, apart from the factor of
 $(-2/\beta)^{p+1}$, as the averaged form of (\ref{3.5L}),
 but now with $\alpha$ replaced by $1/\alpha$ and $N$ replaced by $-N / \alpha$.
 The functional equation (\ref{fa}) then follows. \hfill $\square$

 \subsection{The $(q,t)$ generalised moments --- proof of Corollary \ref{P5}}\label{S4.4}
 The moments of interest are the average (\ref{mbq}).
The generalisation of (\ref{3.5J}) as it applies to
(\ref{mbq}) is known \cite{Ma95} (see too \cite[Eq.~(3.9)]{Ko96}),
 \begin{equation}\label{3.5M}
 \sum_{l=1}
^N x_l^j = \sum_{\kappa: |\kappa| = j}
u_\kappa(q,t) P_\kappa(\mathbf x;q,t), \qquad
u_\kappa(q,t) := (1 - q^{|\kappa|} ) {\chi_\kappa(q,t) \over h_\kappa'(q,t)},
\end{equation}
where $h_\kappa'(q,t)$ is given by (\ref{h1dx}) and
 \begin{equation}\label{chx}
\chi_\kappa(q,t) = \prod_{\substack{s=(i,j) \in \kappa
\\ s \ne (1,1)} }(t^{i-1} - q^{j-1}).
\end{equation}
Taking averages of both sides of (\ref{3.5M}), making use of (\ref{GFq}) on the right hand side, reveals the structure
 \begin{equation}\label{stm}
 \bigg \langle \sum_{l=1}^N x_l^j \bigg \rangle^{(a,q,t)}_N = \sum_{l=0}^j \alpha_l(a,q,t) 
 t^{Nl},
 \end{equation}
 where each $\alpha_l(a,q,t)$ is a rational function of
 $q,t$ and polynomial in $a$. In particular, this functional form justifies the notation $M_j(a,q,t,u)$ with $u = t^N$ as introduced in (\ref{mbq}).

{\it Proof of Corollary \ref{P5}.} 
From the definition (\ref{chx}) we have 
$\chi_{\kappa'}(t,q) = (-1)^{|\kappa| - 1} \chi_{\kappa'}(q,t)$, while from the definitions
(\ref{h1}) and
 (\ref{h1dx}) we have $h_\kappa'(q,t) = h_{\kappa'}(t,q)$. Using these relations in the definition of $u_\kappa(q,t)$ in (\ref{3.5M}) shows
 \begin{equation}\label{3.6M}
 {u_\kappa(q,t) \over h_\kappa(q,t)} = 
  (-1)^{|\kappa|-1}   \bigg ( {1 - q^{|\kappa|} \over 1 - t^{|\kappa|}} \bigg )
 {u_{\kappa'}(t,q) \over h_{\kappa'}(t,q)}.
 \end{equation}
  On the right hand side, by replacing 
 $(q,t) \mapsto (t^{-1}, q^{-1})$ in the first equation of (\ref{3.5M}) and using the first equation is (\ref{4.6a}) we can deduce $u_{\kappa'}(t,q) = u_{\kappa'}(t^{-1},q^{-1})$. Also, from the definitions (\ref{h1}) and (\ref{n1}) we read off that
 \begin{equation}\label{h1p}
 h_\kappa(q,t) = (-1)^{|\kappa|} q^{n(\kappa')}
 t^{n(\kappa) + |\kappa|} h_\kappa(q^{-1},t^{-1}).
 \end{equation}
 Using these facts allows the rewrite of (\ref{3.6M})
 \begin{equation}\label{3.6p} 
 {u_\kappa(q,t) \over h_\kappa(q,t)} = -t^{-n(\kappa)}
  q^{-n(\kappa')-|\kappa|}
     \bigg ( {1 - q^{|\kappa|} \over 1 - t^{|\kappa|}} \bigg )
 {u_{\kappa'}(t^{-1},q^{-1}) \over h_{\kappa'}(t^{-1},q^{-1})}.
  \end{equation}

 To proceed further,
 in analogy to (\ref{3.5L}) we now write (\ref{MA}) in the form
\begin{equation}\label{3.6L}
\sum_{l=1}^N x_l^j =  \sum_{\kappa: |\kappa| = j} 
t^{n(\kappa)} { u_\kappa(q,t) \over h_\kappa(q,t)} 
  {P_\kappa\Big ( \Big \{ p_k = {1 -t^{kN} \over 1 - t^k} \Big \};q,t \Big ) \over 
P_\kappa\Big ( \Big \{ p_k = {1 \over 1 - t^k} \Big \};q,t \Big )} 
{P_\kappa(\mathbf x;q,t) \over
 P_\kappa(1,t,\dots,t^{N-1};q,t)},
 \end{equation} 
 where use has been made of (\ref{8.2}) and (\ref{3.17}).
Next, we average over $\mathbf x$
according to $\langle \cdot \rangle^{(a,q,t)}_N$, making use of (\ref{GFq}) on the right hand side. Further, we substitute for the first ratio in the summand (\ref{3.6L}) according to (\ref{3.6p}), while for the second ratio in the summand we make use of (\ref{MA3}) with
$a^k = - t^{Nk}$. This shows
 \begin{multline}\label{3.7L}
\bigg \langle \sum_{l=1}^N x_l^j \bigg \rangle^{(a,q,t)}_N =  
-   q^{-j} \bigg ( {1 - q^j \over 1 - t^j} \bigg )\sum_{\kappa: |\kappa| = j} q^{n(\kappa')}
 {u_{\kappa'}(t^{-1},q^{-1}) \over h_{\kappa'}(t^{-1},q^{-1})}
{  P_{\kappa'} \Big ( \Big \{ p_k = {1 - t^{Nk}  \over 1 - q^{-k}}
\Big \};t^{-1},q^{-1} \Big ) \over P_{\kappa'} \Big ( \Big \{ p_k = {1 \over 1 - q^{-k}}
\Big \}; t^{-1},q^{-1} \Big )} \\ \times
\bigg \langle 
 {P_{\kappa'}(\mathbf x;t^{-1},q^{-1}) \over
 P_{\kappa'}(1,q^{-1},\dots,q^{-N+1};t^{-1},q^{-1})} \bigg \rangle^{(a,q,t)}_N.
 \end{multline} 
 In this we know that the average on the right hand side is independent of $N$ for $N \ge \ell(\kappa')$. Comparing with the average of (\ref{3.6L}) gives 
 (\ref{mbq1}). \hfill $\square$

 A listing of the Macdonald polynomials for partitions of size up to three, written in terms of the power sum basis, can be found in \cite[Page 6, after correction of the case with $\kappa = (1)^3$]{Sh15} and \cite[Page 50]{Lo21}. From such listings, and together with 
 (\ref{3.5M}) and Theorem \ref{T1}, we can compute for the first three of the moments defined by (\ref{mbq}),
 \begin{align}
 M_1(a,q,t,u) & = (1 + a)  {(1 - u) \over (1 - t)}, \\
 M_2(a,q,t,u) & =  {(1 - u) \over t(1 - t^2)}
 \bigg ( (1 + a^2)t + u\Big ( t + a (1 + q + (1+a+q)t )\Big ) \bigg ),  \label{4.30} \\
 \begin{split}
  M_3(a,q,t,u) & = (1 + a) {(1 - u) \over t^2(1 - t^3)}
  \bigg  (a (1+t)(1 + q + q^2) u^2 + 
  t^2 ((1+a^2)(1 + u + u^2 ) \\ & \qquad + a (-1 + q u) (1 + u + q u) ) \bigg ),    
 \end{split}
  \end{align}
  where $u:= t^N$. Each of these can  be checked to indeed obey the functional relation (\ref{mbq1}). 
  Dividing by $(1-t)$ and setting $a=-1$, we can check that taking the limit $q\to 1$ in (\ref{4.30}) reclaims the evaluation formula for $2 m_{N,2}^{{\rm G} \beta {\rm E}}$ in (\ref{3.6}).
  We remark too on the emergence of obvious structural patterns, the mechanism for which is not  immediate from our formalism (apart from the factor of $(1+a)$ for the odd moments which is required since the weight $w_U^{(a)}(x;q)$ is even for $a=-1$, implying that the odd moments then vanish).

\begin{remark}  
Suppose $t=q$, which is the case $m=1$ of (\ref{5.3}). Then (\ref{mbq1}) can be written
$$
q^{p/2}  M_p(a,q,q,u) = - \Big ( q^{p/2}  M_p(a,q,q,u) \Big ) \Big |_{q \mapsto 1/q}.
$$
Now setting $q = e^{-\lambda/N}$ and thus $u = e^{-\lambda}$, it follows from this that $q^{p/2}  M_p(a,q,q,u)$ permits an expansion in odd powers of $1/N$, generalising the topological expansion  (\ref{3.1}) of the GUE moments. With $a=-1$, this feature has been observed previously in \cite{FLSY23}, and moreover the explicit functional form of the first two terms of the expansion have been computed in \cite{BFO24}. 
\end{remark}

%\subsection*{Data availability statement} There is no data associated to this work.

%\subsection*{Conflict of interest statement} The authors have no conflicts of interest to disclose.

%\bibliographystyle{amsplain}
%\bibliography{book1t16}
\nopagebreak

\providecommand{\bysame}{\leavevmode\hbox to3em{\hrulefill}\thinspace}
\providecommand{\MR}{\relax\ifhmode\unskip\space\fi MR }
% \MRhref is called by the amsart/book/proc definition of \MR.
\providecommand{\MRhref}[2]{%
  \href{http://www.ams.org/mathscinet-getitem?mr=#1}{#2}
}
\providecommand{\hyperref}[2]{#2}

\end{document}